\documentclass[10pt,letterpaper, twocolumn]{revtex4}
\usepackage[T1]{fontenc}
\usepackage{amsmath}
\usepackage{amsthm}
\usepackage{amsfonts}
\usepackage{graphicx}
\usepackage{dsfont}

\newcommand{\ket}[1]{|#1\rangle}
\newcommand{\bra}[1]{\langle#1|}
\newcommand{\ev}[1]{\langle#1\rangle}
\newcommand{\dg}[1]{#1^\dagger}
\newcommand{\phdg}[1]{#1^{\phantom\dagger}}
\newcommand{\trans}[1]{#1^{\text{T}}}
\newcommand{\mc}{\mathcal}

\DeclareMathOperator{\tr}{tr}
\DeclareMathOperator{\argtanh}{argtanh}

\newtheorem{theorem}[]{Theorem}
\newtheorem{lemma}[]{Lemma}

\newtheorem{definition}[]{Definition}

\begin{document}

\title{Gaussian matrix product states cannot efficiently describe critical systems}
\author{Adri\'an Franco-Rubio}
\affiliation{Max-Planck-Institut für Quantenoptik, Hans-Kopfermann-Stra\ss e 1, 85748 Garching, Germany}
\affiliation{Munich Center for Quantum Science and Technology, Schellingstra\ss e 4, 80799 München, Germany}
\author{J. Ignacio Cirac}
\affiliation{Max-Planck-Institut für Quantenoptik, Hans-Kopfermann-Stra\ss e 1, 85748 Garching, Germany}
\affiliation{Munich Center for Quantum Science and Technology, Schellingstra\ss e 4, 80799 München, Germany}
\date{\today}

\begin{abstract}
    Gaussian fermionic matrix product states (GfMPS) form a class of ansatz quantum states for 1d systems of noninteracting fermions. We show, for a simple critical model of free hopping fermions, that: (i) any GfMPS approximation to its ground state must have bond dimension scaling \emph{superpolynomially} with the system size, whereas (ii) there exists a non-Gaussian fermionic MPS approximation to this state with \emph{polynomial} bond dimension. This proves that, in general, imposing Gaussianity at the level of the tensor network may significantly alter its capability to efficiently approximate critical Gaussian states. We also provide numerical evidence that the required bond dimension is subexponential, and thus can still be simulated with moderate resources.
\end{abstract}

\maketitle

The growing complexity of quantum many-body wavefunctions with increasing system sizes has motivated the development of variational classes of states. By exploiting simplifying features of a given problem, ansatz states can help optimize numerical resources, as well as provide an insightful new perspective into the inner workings of quantum correlations in these systems. For instance, \textit{Gaussian} states and their correlation matrix formalism greatly facilitate computations involving non-interacting particles. On the other hand, \textit{tensor network} states have become an essential theoretical framework and numerical toolbox for quantum many-body physics, excelling at the representation of area law and similarly low-entangled states \cite{Cirac21}.

For systems of free (or weakly interacting) fermions, both classes can be combined to give rise to \textit{Gaussian fermionic tensor network} states \cite{Kraus10}. Relevant examples include the ground state of the Kitaev-Majorana chain, and models of topological insulators and superconductors \cite{Wahl13, Dubail15, Wahl14, Yang15}. In the 1d case, the resulting tensor network is the \textit{Gaussian fermionic matrix product state} or GfMPS, which has been shown to outperform non-tensor network based methods in free fermion computations for very large systems \cite{Schuch19}. This motivates the question about the expressivity of GfMPS, namely what kind of free fermionic states can be efficiently described by them.

In the case of general MPS, it was proved in \cite{Verstraete06} that an efficient approximation (i.e. one with bond dimension growing at most polynomially with the system size $N$) exists whenever a certain Rényi entropy is bounded by $O(\log N)$. This established the usefulness of MPS to approximate states with at most a logarithmic violation of the area law, including the ground states of both gapped and gapless (critical) local Hamiltonians. In the setting of fermionic chains, it also applies to fermionic MPS (fMPS). However, it is not known whether an analogous result holds for GfMPS whenever the state being approximated is Gaussian with similarly bounded entropies.

Here we answer this question in the negative: we provide a simple counterexample in the form of a critical hopping fermion Hamiltonian, whose Gaussian ground state can be efficiently approximated by fMPS but not by GfMPS. The proof of this last fact combines  a rigorous bound on the error incurred by a fixed rank Gaussian approximation of a Gaussian state with specific knowledge of the entanglement structure of the target ground state, obtained from asymptotic Toeplitz determinant theory. Furthermore, we provide evidence, both from conformal field theory arguments and numerical results, that the required bond dimension scaling for a good GfMPS approximation is nevertheless subexponential. This makes the question about the existence of an efficient GfMPS approximation a hard one to settle numerically, which motivated our pursuit of an analytical proof.

Our result disproves the somewhat intuitive assumption that the most bond dimension-efficient approximation to a Gaussian state would come from a Gaussian tensor network. This can be relevant when optimizing resources for computational applications, though it should be noted that the savings due to having access to the correlation matrix formalism may well compensate the extra bond dimension derived from Gaussianity. Additionally, there are other situations where Gaussian tensor networks have faced difficulty approximating Gaussian states, as is the case for ground states of local, gapped, quadratic Hamiltonians displaying chiral topological features \cite{Wahl13, Dubail15}. In this context, our findings leave the door open to the existence of better, non-Gaussian tensor network approximations that bypass the no-go results.

\textit{Model}--- We consider a periodic chain of length $N$ with a single fermionic mode $a_i, a_i^\dagger$ per site, satisfying the usual canonical anticommutation relations,
\begin{equation}
    \{\phdg a_i,\dg a_{j}\} = \delta_{ij},\qquad \{a_i,a_j\} = \{\dg a_i, \dg a_j\} = 0,
\end{equation}
and study the free hopping Hamiltonian at half filling,
\begin{align}
    	H = -\dfrac{1}{2}\sum_{j=1}^N{\dg a_j \phdg a_{j+1}} + \text{h.c.}= -\sum_{k}{\cos{k}~\dg a_k \phdg a_k},
    	\label{Ham}
\end{align}
where we have defined the momentum modes in the usual form, ${a_k \equiv \frac{1}{\sqrt{N}}\sum_{j=1}^N{e^{ikj}a_j}}$ with $k\in\frac{2\pi}{N}\mathbb{Z}\cap(-\pi,\pi]$. In this basis $H$ is diagonal and its ground state, which is Gaussian due to $H$ being quadratic, can be determined by filling the negative energy modes below the Fermi momentum, $k_F = \pi/2$ \footnote{Whenever $N\equiv 0 \mod 4$, there are zero modes $a_{\pm k_F}$ sitting at the Fermi points, so that the ground state is four-fold degenerate. This does not affect our results so we will ignore this circumstance in our discussion.}. This is encoded in the momentum space correlation matrix,
\begin{equation}
    C_{kq} \equiv \ev{\dg a_k \phdg a_q} = n_k\delta_{k,q},\qquad n_k\equiv\Theta(k_F-|k|),
    \label{corrmat}
\end{equation}
where $\Theta$ denotes the Heaviside step function. The position space correlation matrix $C_{ij}\equiv \ev{\dg a_i \phdg a_j}$ can then be obtained by an inverse Fourier transform. It exhibits power-law decays as befits a gapless model (its explicit form can be seen in Appendix \ref{app:thm_proof}). Note that due to particle number conservation, $\ev{a_ka_q} = \ev{a_ia_j} = 0$.

The entanglement structure of this state, which will be key to the results presented next, can be obtained from its correlation matrix. Given a bipartition of a pure Gaussian state into complementary regions $\mc R, \bar{\mc  R}$, we can find a basis of modes on each subsystem such that the state decomposes as the tensor product of entangled fermion pairs \cite{Botero04}. How entangled these pairs are is given by the spectrum of the correlation matrix of either subsystem, which is nothing but the corresponding submatrix of the global correlation matrix,
\begin{equation}
    C_{\mc R}\equiv (C_{ij})_{i,j\in \mc R}.
\end{equation}
For convenience and notational unity we will work with the eigenvalues of $V_{\mc R}\equiv2C_{\mc R}-\mathds{1}$, which we denote ${\{\lambda_j\}\subset[-1,1]}$, and call the \textit{Gaussian entanglement spectrum} \footnote{Whenever $\mc R$ is (say) larger that its complement, some of the modes from $\mc R$ will not be entangled to modes in $\bar{\mc R}$, but rather remain in a product state, their contribution to the entanglement spectrum being $\lambda=\pm1$.}. The Rényi entropy $S_\alpha$ then splits as a sum of contributions from each entangled pair,
\begin{equation}
S_{\alpha} =\sum_j{s_\alpha(\lambda_j)},
\end{equation}
where
\begin{equation}
    s_\alpha(\lambda)\equiv\dfrac{1}{1-\alpha}\log\left[\left(\dfrac{1+\lambda}{2}\right)^\alpha+\left(\dfrac{1-\lambda}{2}\right)^\alpha\right],
\end{equation}
so that the entanglement decreases with $|\lambda|$ from $\lambda = 0$ (maximally entangled state) to $\lambda=\pm1$ (product state).

For the ground state of $H$, the leading scaling of the Rényi entropy of an interval of size $L$ can be seen to be logarithmic \cite{Vidal03, Jin04, Peschel09},
\begin{equation}
    S_\alpha(L) \sim \dfrac{\alpha+1}{6\alpha}\log L,\qquad L \to\infty,
    \label{ent_sca}
\end{equation}
which is consistent with it lying in the universality class of the free boson conformal field theory (CFT) with central charge $c=1$ \cite{Holzhey94, Calabrese04}.

\textit{Efficient approximation with fMPS} --- A fermionic matrix product state (fMPS) \cite{Kraus10} is defined in terms of a series of so-called \textit{fiducial} states of $f$ physical fermions and $2\chi$ virtual fermions. The state represented by the fMPS is obtained by \textit{contracting} the virtual fermions, i.e. projecting them onto maximally entangled pairs (see \mbox{Fig. \ref{fig:fMPS}}). The dimension $D$ of the virtual Hilbert space, i.e. the bond dimension (b.d.) of the fMPS, is related to $\chi$ via $D = 2^\chi$ \footnote{It is not uncommon to work in the language of Majorana operators, in which case $\chi$ is usually redefined to denote the number of virtual Majorana operators. If, additionally, we are working with periodic boundary conditions or infinite systems, it is possible to have an \textit{odd} number of virtual Majorana operators. This is the case for the GfMPS that we used in our numerics (see Appendix \ref{app:num}). For our theoretical proof, however, we will restrict ourselves to open boundary conditions (OBC), since periodic MPS can always be recast in OBC form.}.

\begin{figure}[t]
	\centering
	\includegraphics[width=.9\linewidth]{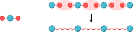}
	\caption{(left) A (G)fMPS fiducial state of one physical fermion and two virtual fermions. (right) A (G)fMPS is obtained from projecting the virtual fermions onto entangled pairs, leaving an entangled state of the physical fermions.}
	\label{fig:fMPS}
\end{figure}

\begin{definition}
A family of states $\ket{\Psi_N}$ for increasing system sizes $N$ is \emph{efficiently approximable} by fMPS if for any $\epsilon>0$ there exists a family of fMPS states $\ket{\Psi^{\text{\tiny MPS}}_N}$ with b.d.~$D(N)=\text{\emph{poly}}(N)$ and
\begin{equation}
    \|\ket{\Psi_N} - \ket{\Psi^{\text{\tiny MPS}}_N}\|_2\leq\epsilon\quad \forall N.
\end{equation}
\end{definition}

\begin{theorem}
The family of ground states of \eqref{Ham} is efficiently approximable by fMPS.
\end{theorem}
\begin{proof}
Using the Jordan-Wigner transformation, we can map our system to a spin chain model, the XX model. Since we have a logarithmic bound \eqref{ent_sca} for the Rényi entropies, in particular for $\alpha<1$, it follows from Lemmas 1 and 2 in \cite{Verstraete06} that the ground state of the XX model is efficiently approximable by MPS. Then, undoing the Jordan-Wigner transformation, we can find an efficient fMPS approximation for the fermionic model (this is done in detail in Appendix \ref{app:fMPS}).
\end{proof}

In general, understanding the entanglement structure of quantum states is key to obtain both approximability and inapproximability results, since the bond dimension of an MPS has a clear interpretation as the maximum Schmidt rank for a bipartition of an MPS into connected subsystems. In order to compare it with an analogous result in the next section, we cite here the following

\begin{lemma}[Low Schmidt rank approximation]
Let $\ket{\Psi}$ by a bipartite quantum state with Schmidt spectrum $\{s_j\}_{j=1}^n$ in descending order. Then for any bipartite state $\ket{\tilde\Psi}$ of Schmidt rank at most $r$,
\begin{equation}
    |\bra{\Psi}\tilde\Psi\rangle|^2\leq 1-\sum_{j=r+1}^n{s_j^2},
\end{equation}
and the bound is tight: the optimal $\ket{\tilde\Psi}$ can be found by truncating the Schmidt decomposition of $\ket{\Psi}$.
\label{EYM}
\end{lemma}
Lemma \ref{EYM} is a consequence of the Eckart-Young-Mirsky theorem \cite{Eckart36, Mirsky60}, which states that the optimal low rank approximation to a given matrix comes from truncating its singular value decomposition. It leads to a \textit{lower} bound for the error of an MPS approximation \footnote{Even though, we stated the definition of approximability in terms of norms (as it is customary) we will, for convenience, mostly work in terms of fidelities, as in Lemmas \ref{EYM} and \ref{FGRA}. This does not change anything since both approaches are equivalent.} (used for instance in some of the inapproximability proofs in \cite{Schuch08}).

\textit{No efficient approximation with Gaussian fMPS} --- A Gaussian fMPS (GfMPS) is an fMPS for which all fiducial states are Gaussian. Since the contraction operation projects each pair of virtual modes onto a Gaussian state, the maximally entangled pair, the global state after contraction is necessarily Gaussian. The contraction of GfMPS tensors can be done at the level of correlation matrices via Schur complements \cite{Bravyi05, Schuch19} (see Appendix \ref{app:num}). 

An efficient approximation in terms of GfMPS can be defined analogously to the fMPS case. Then, we have our main result as
\begin{theorem}
The family of ground states of \eqref{Ham} is \emph{not} efficiently approximable by GfMPS.
    \label{thm_superpoly}
\end{theorem}

The proof of Theorem \ref{thm_superpoly} will follow from two lemmas. The first one is the Gaussian version of Lemma \ref{EYM}. Given the Gaussian entanglement spectrum $\{\lambda_j\}$ of a bipartite state, we call the number of eigenvalues $\lambda_j\neq \pm1$ its \textit{Gaussian rank}. We then have
\begin{lemma}[Low Gaussian rank approximation]
\label{FGRA}
Let $\ket{\Psi}$ be a bipartite Gaussian state, with Gaussian entanglement spectrum $\{\lambda_i\}_{i=1}^n$, ordered so that $|\lambda_i|\leq|\lambda_{i+1}|$. Then for any Gaussian state $\ket{\tilde\Psi}$ of Gaussian rank at most $r$,
\begin{equation}
   |\bra{\Psi}\tilde\Psi\rangle|^2 \leq \prod_{i=r+1}^n{\dfrac{1+|\lambda_i|}{2}}=\exp{\left(-S_\infty^{\emph{trunc}}[r]\right)},
    \label{maxoverlap}
\end{equation}
where $S_\infty^{\emph{trunc}}[r] \equiv \sum_{i=r+1}^n{s_\infty(\lambda_i)}$, and the bond is tight: the optimal $\ket{\tilde\Psi}$ can be found by truncating the Gaussian singular value decomposition of $\ket{\Psi}$.
\end{lemma}
We were not able to find a proof in the literature so we provide one together with related results in Appendix \ref{app:FGRA}. Lemma \ref{FGRA} can be used to lower bound the error of a GfMPS approximation to a given state, since the Gaussian rank of any GfMPS divided into two connected subsystems is upper bounded by $\chi$. We then need information on the entanglement spectrum of our target state, which is provided by

\begin{lemma}
For the ground state of \eqref{Ham}, let $\mathcal{I}_{L,N}(\mu)$ be the number of eigenvalues $\lambda$ from the Gaussian entanglement spectrum of an interval of size $L$ in a chain of $N$ sites that satisfy $|\lambda|<\mu$, and let $c>0$. Then there exists $\mu<1$ such that
\begin{equation}
    \mathcal{I}_{L,N}(\mu) > c\log N,
\end{equation}
as $L,N\to\infty$ with $L/N$ fixed.
\label{property}
\end{lemma}

The proof of Lemma \ref{property} can be found in Appendix \ref{app:thm_proof}. It starts by proving the equivalent property for the Gaussian entanglement spectra of the infinite chain: in the thermodynamic limit, we can exploit the theory of Toeplitz determinants to lower bound the corresponding $\mathcal{I}_{L,\infty}(\mu)$ function in the asymptotic regime. Then we use standard inequalities to show that the difference between the finite and infinite chain correlation matrices is bounded in trace norm. This ensures that their respective spectra are distributed similarly enough so that Lemma \ref{property} follows.

\begin{proof}[Proof (of Thm. \ref{thm_superpoly})]
Suppose there exists a GfMPS approximation with polynomial b.d.~$D(N)$. Then we can find ${c>0}$ such that ${\chi(N)=\log_2{D(N)}\leq c\log N}$. Thanks to Lemma \ref{property}, we know there exists some ${\mu<1}$ such that ${\mathcal{I}_{L,N}(\mu) > (c+1)\log N}$, and we have
\begin{align}
    S_\infty^{\emph{trunc}}[\chi(N)]&\geq (\mathcal{I}_{L,N}(\mu) - \chi(N))s_\infty(\mu)\nonumber\\
    &\geq s_\infty(\mu)\log N,
\end{align}
which diverges as $N\to\infty$. Since $\ket{\tilde\Psi_N}$ has Gaussian rank bounded by $\chi(N)$ across the bipartition, Lemma \ref{FGRA} implies that the overlap between the ground state and its GfMPS approximation goes to zero as the system size increases. Thus, by contradicion, any approximation with bounded error must have $\chi(N)$ growing faster than logarithmically, and consequently $D(N)$ grows superpolynomially.
\end{proof}

\textit{CFT argument} --- The techniques used in the proof of Theorem \ref{thm_superpoly} cannot be applied to obtain better lower bounds, or upper bounds on the required b.d., for which more accurate knowledge of the Gaussian entanglement spectra of finite chains would be needed. Here we provide evidence that this b.d.~is subexponential. First, a heuristic argument is made, based on conformal field theory (CFT). In the next section, we present some numerical results. 

The low-lying entanglement spectrum of a critical model is know to behave universally according to the underlying CFT. With our notation, for an interval of $L$ sites in an chain of $N$ sites, we have \cite{Peschel04, Lauchli13, Ohmori15, Cardy16}
\begin{equation}
    |\lambda_n| \simeq \tanh{\left(\dfrac{\pi^2}{2}\dfrac{\varepsilon_n}{\log \ell}\right)},\qquad \ell = \dfrac{N}{\pi a}\sin{\dfrac{\pi L}{N}} 
    \label{cftspec}
\end{equation}
where $\ell$ is the effective length of our interval in units of some UV cutoff $a$, and the $\varepsilon_n$ are fixed by the CFT. In our model, this is the free compactified boson CFT, and
\begin{equation}
    \varepsilon_n \equiv \left\lfloor\dfrac{n}{2}\right\rfloor+\dfrac{1}{2} = \dfrac{1}{2}, \dfrac{1}{2}, \dfrac{3}{2}, \dfrac{3}{2}, \dfrac{5}{2}, \dfrac{5}{2}\ldots
\end{equation}
(More precisely, these numbers characterize the spectrum of scaling dimensions of an associated boundary CFT.) The CFT spectrum \eqref{cftspec} not only satisfies the condition in Lemma \ref{property}, but it also allows us to estimate the tail contribution to the $\infty$-Rényi entropy,
\begin{equation}
    S_\infty^{\text{trunc}}[\chi] \approx\dfrac{2\log \ell}{\pi^2}\exp\left(-\dfrac{\pi^2\chi}{2\log \ell}\right).
\end{equation}
Thus, if the CFT prediction were exact for the whole spectrum, by Lemma \ref{FGRA} the required scaling for $\chi$ is
\begin{equation}
    \chi(N,\epsilon)\approx\dfrac{2}{\pi^2}\log{\eta N} \log\left(\dfrac{\pi^2}{2\epsilon}\log{\eta N}\right),
    \label{heuristic_scaling}
\end{equation}
for some proportionality constant $\ell = \eta N$. Here we used the fidelity error $\epsilon\equiv 1-\exp{\left(-S_\infty^{\emph{trunc}}[\chi]\right)}$ attached to a single bipartition of the system. Our experience from MPS is that we should account for all bipartitions, with different $L/N$ and thus different $\eta$ \cite{Verstraete06}. We do this coarsely by replacing $\epsilon\to\epsilon/N$. This results in
\begin{align}
    D(N,\epsilon) \approx (\eta N)^{\frac{\log 2}{\pi^2}\log{\left(\frac{2N\log \eta N}{\pi^2 \epsilon}\right)}},
    \label{heuristic_scaling_N}
\end{align}
which is subexponential.

\textit{Numerics} --- We have also performed some numerical studies of this problem which seem to partially confirm the scaling from \eqref{heuristic_scaling}. We introduce them briefly here, and elaborate on them in Appendix \ref{app:num}. Essentially, we defined a subclass of translation invariant GfMPS (which we call ladder GfMPS), which exactly represent states with the following occupation number in momentum space:
\begin{equation}
    n_k = \dfrac{p(\cos \frac{k}{2})^2}{p(\cos \frac{k}{2})^2+q(\sin \frac{k}{2})^2},
    \label{generic_nk}
\end{equation}
for $p,q$ arbitrary real odd monic polynomials. Clearly, to reproduce \eqref{corrmat}, we need to have $p$ (resp.~$q$) supported mostly inside (resp.~outside) the Fermi surface. We tried several choices, the best of which is represented in Fig.~\ref{fig:numer}. There we have used $\delta \equiv \bra{\Psi^{\text{\tiny MPS}}}H_{\text{fb}}\ket{\Psi^{\text{\tiny MPS}}} - E_0$, with $H_{\text{fb}}$ the flat band Hamiltonian with the same ground state as $H$, and $E_0$ its ground state energy, as a proxy for the fidelity error $\epsilon$, which it upper bounds. Further energy optimization (as done in \cite{Mortier20}), which we did not pursue, could improve the results. Fig.~\ref{fig:numer} also shows the estimation \eqref{heuristic_scaling_N} for $\eta\sim1.3$ (resulting from optimization), which gives an idea of the scaling of our numerical results.
 
\begin{figure}[t]
	\centering
	\includegraphics[width=1\linewidth]{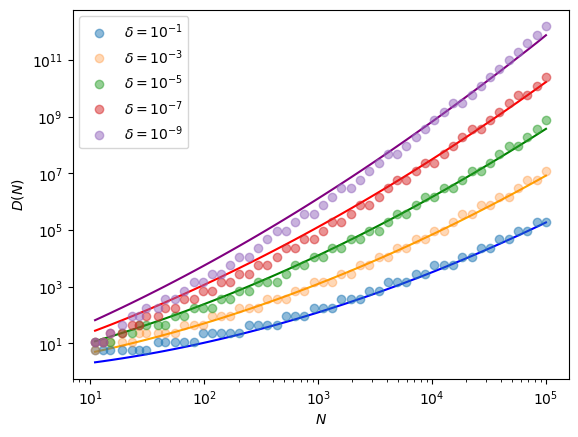}
	\caption{Bond dimension $D(N)$ vs. system size $N$ for our (not necessarily optimal) ladder GfMPS approximation to the ground state of $H$. $\delta$ is the error in energy, which upper bounds the fidelity error $\epsilon$. The lines represent \eqref{heuristic_scaling_N} with numerically optimized $\eta\sim1.3$.}
	\label{fig:numer}
\end{figure}

\textit{Toy model} --- Finally, we present a simple toy model of a family of Gaussian states, with no reference to a Hamiltonian, that are efficiently approximable by fMPS but not by GfMPS. We begin by explaining the intuition behind it. Lemmas \ref{EYM} and \ref{FGRA} highlight the differences between Gaussian and non-Gaussian truncation of a bipartite state \footnote{In fact, these lemmas correspond to the fMPS/GfMPS approximation problem in its simplest form, i.e. when $N=2$.}. This dichotomy is also reflected in different bounds for the truncation errors, $\epsilon_G$ and $\epsilon_{NG}$ resp., in terms of the $\alpha<1$ Rényi entropy. Assuming $\epsilon_G,\epsilon_{NG}\ll 1$, we have,
\begin{align}
    \epsilon_G &\lesssim S_\alpha^{1/\alpha} \chi^{-(1-\alpha)/\alpha},\\
    \epsilon_{NG} &\lesssim S_\alpha^{1/\alpha} D^{-(1-\alpha)/\alpha}.
\end{align}
For the non-Gaussian case this is Lemma 2 in \cite{Verstraete06}. The Gaussian case can be proved similarly by minimizing the entropy over all possible Gaussian entanglement spectra with a fixed Gaussian truncation error. Since $\chi\sim\log D$, we see that the bond for $\epsilon_G$ is much weaker, i.e. there exist states with small $S_\alpha$ and $\epsilon_{NG}$ but high $\epsilon_G$ for the same bond dimension \footnote{It should be noted that the $\alpha>1$ entropy bounds, such as those in \cite{Schuch08}, are the same for standard and Gaussian truncation.}. These states are usually characterized by Gaussian entanglement spectra involving many low-entangled pairs. 

We exploit this in our toy example, leveraging the addition of entangled pairs as we grow the system by making their entanglement increasingly weaker. Consider a family of states $\ket{\psi_N}$ on rings of $N$ sites obtained by distributing $\nu(N)$ entangled pairs between opposite sites in the chain, each with strength $\lambda_N$, where
\begin{equation}
    \nu(N) \sim (\log N)^{1+\beta},\quad \lambda_N \sim 1 - \dfrac{2}{\nu(N)},
\end{equation}
for some $\beta>0$. The $\alpha<1$ Rényi entropy is maximal for a bipartition of the ring into two equal halves, since all pairs contribute to give
\begin{align}
    S_\alpha(N/2) &= \nu(N) s_\alpha(\lambda_N) \sim (\log N)^{(1-\alpha)(1+\beta)},
\end{align}
which is upper bounded by $O(\log N)$ for $\alpha\in\left(\frac{\beta}{1+\beta},1\right)$. Therefore, for any $\beta>0$ an efficient fMPS approximation exists. However, for any GfMPS approximation $\ket{\tilde\psi_N}$, the error bound \eqref{maxoverlap} reads
\begin{align}
    |\bra{\psi_N}\tilde\psi_N\rangle|&\leq\left(\dfrac{1+\lambda_N}{2}\right)^{\nu(N)-\chi(N)}\\
    &\lesssim\left(1-\dfrac{1}{\nu(N)}\right)^{\nu(N)-\chi(N)}.
\end{align}
Thus if $|\bra{\psi_N}\tilde\psi_N\rangle|\geq 1-\epsilon$, then for large enough $N$ and small enough $\epsilon$ we have
\begin{align}
    \chi(N)\gtrsim \left(1-\epsilon\right)\nu(N)\sim (\log N)^{1+\beta},
\end{align}
so that $D(N)\gtrsim N^{(\log N)^{\beta}}$. This scaling can be easily seen to be sufficient (since $\chi(N)=\nu(N)$ allows for an exact representation), so that the required b.d.~to approximate $\ket{\psi_N}$ is superpolynomial but still subexponential.

\textit{Acknowledgments} --- A.F.R. would like to thank F.~Ares for introducing him to Toeplitz asymptotics a few years ago. A.F.R. is supported by the Alexander von Humboldt Foundation. I.C. acknowledges funding from the ERC Grant QUENOCOBA (No 742102) and the DFG through the DACH Lead Agency Agreement Beyond C (No 414325145).

\newpage
\appendix

\section{fMPS approximations from (spin) MPS approximations}
\label{app:fMPS}
Following \cite{Kraus10}, a fermionic MPS with $f = 1$ physical fermions and $2\chi = 2$ virtual fermions per site can be defined by means of a series of ``projectors''
\begin{equation}
    Q_j \equiv \sum_{\substack{k, l, m~=~0,1\\k\oplus l\oplus m=p}}{[A_j]^k_{lm}~(\dg a_j)^k~\alpha_j^l~\beta_j^m},
\end{equation}
where $\oplus$ denotes the sum in $\mathbb{Z}_2$, $j$ is the position index along the chain, $[A_j]^k_{lm}$ is a coefficient tensor, and  $a_j, \alpha_j,\beta_j$ are the physical, left virtual and right virtual fermionic modes at site $j$, respectively. The restriction on the indices ensures a well-defined fermionic parity $p$ for $Q_j$. One also needs to define the operators
\begin{equation}
    H_j = \dfrac{1}{\sqrt{2}}\left(\mathds{1}+ \dg \beta_j\dg \alpha_{j+1}\right),
\end{equation}
which generate entangled pairs of virtual fermions from the vacuum. The fMPS state is then defined by the action of both sets of operators on the global vacuum, followed by projecting out the virtual fermions:
\begin{equation}
    \ket{\Psi} \equiv \bra{0}_{\text{virtual}}\prod_j{Q_j}\prod_j{H_j}\ket{0}_{\text{all}}.
\end{equation}
The generalization to larger physical and/or bond dimensions is straightforward. If we now interpret the fermionic Fock space as the Hilbert space of a spin $\frac{1}{2}$ chain (as is done implicitly by the Jordan-Wigner transformation), one can see that $\ket{\Psi}$ can be obtained from a standard ``spin'' MPS, with local tensors that coincide with $[A_j]^k_{lm}$, possibly up to signs. Consequently, the set of 1d fMPS states coincides with those obtained from MPS that satisfy: (i) their bond dimension is a power of 2, and (ii) each tensor has a well-defined parity.

In the case we are interested in, the Jordan-Wigner transformation maps the fermionic Hamiltonian $H$ from \eqref{Ham} to the XX model spin chain,
\begin{equation}
    H_{XX} = \dfrac{1}{2}\sum_{j}{\left[X_jX_{j+1} + Y_jY_{j+1}\right]},
\end{equation}
for which the results from \cite{Verstraete06} apply, since all its Rényi entropies grow logarithmically (see Eq. \eqref{ent_sca}). The existence of an MPS approximation with polynomial bond dimension then follows. Additionally, this spin MPS has a global parity symmetry represented by the action of the product of $Z$ operators on each physical spin. To see that this implies the existence of an fMPS approximation to the ground state with polynomial b.d., we show that we can make the spin MPS satisfy conditions (i) and (ii) without a substantial increase in bond dimension. 

Let $D$ be the bond dimension of the spin MPS, and $[B_j]^k_{lm}$ denote its tensors. Then there is $q$ such that $D\leq 2^q<2D$, and we can embed the $D\times D\times 2$-dimensional MPS tensors into $2^q\times 2^q\times 2$-dimensional ones without changing the state just by padding each tensor with zeros, i.e. by extending the range of $l,m$ to $2^q$, and defining the additional tensor elements as 0. Thus we can satisfy (i) with less than twice the original bond dimension. To get condition (ii), we modify the tensors by adding an additional pair of indices $l', m'\in\{0,1\}$ such that
\begin{equation}
[B_j]^k_{ll'mm'}\equiv[B_j]^k_{lm}\delta_{k\oplus|l|\oplus|m|\oplus l' \oplus m'},
\end{equation}
where $|l|,|m|$ denotes the parity of the corresponding index. The new tensors are individually parity symmetric, so that (ii) holds, and they can be seen to generate the same state as the original ones (which is only possible because said state has global parity symmetry). Thus it follows that there exists an fMPS approximation to the ground state of $H$, with less than four times the bond dimension of the spin MPS approximation to the XX model ground state, which therefore grows at most polynomially.

\section{Gaussian bipartite state overlaps and Gaussian entanglement spectrum}
\label{app:FGRA}
In this Appendix we prove Lemma \ref{FGRA} as a corollary to a more general result. For convenience, we work here in the Majorana representation, introducing Majorana operators
\begin{eqnarray}
    &c_{j,1}\equiv \phdg a_j + \dg a_j,\qquad c_{j,2}\equiv i(\phdg a_j-\dg a_j),\\
    &\{c_{j,s},c_{j',s'}\} = 2\delta_{j,j'}\delta_{s,s'}
\end{eqnarray}
so that the (Majorana) correlation matrix is a real, skew-symmetric matrix defined as
\begin{equation}
    [\Gamma]_{js,j's'} \equiv \dfrac{i}{2}\ev{[c_{j,s},c_{j',s'}]}\qquad s,s'= 1,2,
\end{equation}
which satisfies $\Gamma\trans\Gamma\leq \mathds{1}$, with equality for pure states. The overlap of two Gaussian states of $2N$ Majorana fermions can be computed from their correlation matrices $\Gamma, \gamma$, \cite{Bravyi05}
\begin{equation}
\left|\bra{\Gamma}\gamma\rangle\right|^2 = 2^{-N}\sqrt{\det(\mathds{1}-\Gamma\gamma)}.
\label{overlap}
\end{equation}
Finally, we introduce the \textit{Gaussian singular value decomposition} (Gaussian SVD) \cite{Botero04}, which states that, given $\Gamma$ the correlation matrix of a pure bipartite Gaussian state on two subsystems of $2n$ Majorana fermions each \footnote{Once more, in case one of the subsystems is bigger than the other, additional blocks representing the remaining fermions in product states have to be added. This does not affect the rest of the discussion.}, we can find $O,Q\in O(2n)$ that block diagonalize $\Gamma$,
\begin{equation}
    \Gamma = (O\oplus Q)\left(\bigoplus_{j=1}^n{W(\theta_j)}\right)(O\oplus Q)^{\text{T}},
    \label{GaussianSVD}
\end{equation}
where the $4\times 4$ blocks are given by
\begin{equation}
    W(\theta)\equiv \left(\begin{array}{cc}\cos\theta J & \sin\theta \mathds{1} \\ -\sin\theta \mathds{1} & -\cos\theta J\end{array}\right), \quad J \equiv \left(\begin{array}{cc}0 & 1 \\ -1 & 0\end{array}\right),
    \label{defW}
\end{equation}
and the $\theta_j$ can all be chosen to lie on the first quadrant, $0\leq\theta_j\leq\frac{\pi}{2}$, in which case we have $\cos\theta_j = |\lambda_j|$. The $\theta_j$ are another possible way to write the Gaussian entanglement spectrum. Indeed, $W(\theta)$ is the correlation matrix of a pair of fermionic modes, which is in a product state for $\sin{\theta} = 0$ ($|\lambda|=1$) and maximally entangled whenever $\cos{\theta}=0$ ($\lambda=0$). The number $r$ of entangled pairs ($\sin{\theta} \neq 0$) is what we called in the main text the \textit{Gaussian rank}.
Now we are ready to state and prove the following
\begin{theorem}
\label{thm_overlap}
Let $\ket{\Gamma},\ket{\tilde\Gamma}$ be pure bipartite Gaussian states on two subsystems of $2n$ Majorana fermions. Let their correlation matrices be $\Gamma, \tilde\Gamma$ and their Gaussian entanglement spectra be given by $\{\theta_j\}_{j=1}^n,\{\tilde\theta_j\}_{j=1}^n$ respectively.
Then,
\begin{equation}
    \left|\bra{\Gamma}\tilde\Gamma\rangle\right|^2 \leq \max_{\sigma\in S_n}\prod_{i=1}^n{\cos^2\left(\dfrac{\theta_i-\tilde\theta_{\sigma(i)}}{2}\right)},
    \label{maxoverlap_app}
\end{equation}
and the bound is tight (it is reached for some $\Gamma, \tilde\Gamma$).
\end{theorem}
\begin{proof}
We follow a strategy inspired by Thm.~VI.7.1 in \cite{Bhatia96}. $\Gamma, \tilde\Gamma$ will be of the form
\begin{align}
    \Gamma = (O\oplus Q)\left(\bigoplus_{j=1}^n{W(\theta_j)}\right)(O\oplus Q)^{\text{T}},\\
    \tilde\Gamma = (\tilde O\oplus \tilde Q)\left(\bigoplus_{j=1}^n{W(\tilde \theta_j)}\right)(\tilde O\oplus \tilde Q)^{\text{T}},
\end{align}
for some $O, \tilde O, Q, \tilde Q\in O(2n)$.
To begin with, we shall assume that 
\begin{align}
    \theta_i,\tilde\theta_i\in\left(0,\frac{\pi}{2}\right), \qquad & \forall i,\\
    \theta_i \neq \theta_j, \tilde\theta_i \neq \tilde\theta_j, \qquad &  i\neq j.
\end{align}
We are seeking to upper bound
\begin{equation}
    \left|\bra{\Gamma}\tilde\Gamma\rangle\right|^2 = 2^{-2n}\sqrt{\det(\mathds{1}-\Gamma\tilde\Gamma)} = 2^{-2n}\sqrt{\det{(\Gamma + \tilde\Gamma)}},
\end{equation}
where we have used \eqref{overlap} and the purity condition ${\Gamma^{-1} = -\Gamma}$. In other words, our problem consists in determining
\begin{equation}
    \max_{\substack{O, Q, \\\tilde O, \tilde Q}}\det((O\oplus Q)W(O\oplus Q)^{\text{T}} + (\tilde O\oplus \tilde Q)\tilde W (\tilde O\oplus \tilde Q)^{\text{T}})
\end{equation}
where
\begin{equation}
    W \equiv \left(\bigoplus_{j=1}^n{W(\theta_j)}\right),\qquad\tilde W \equiv \left(\bigoplus_{j=1}^n{W(\tilde \theta_j)}\right),
\end{equation}
We know the maximum exists since we are optimizing over a closed manifold. Further, we can assume $O,Q = \mathds{1}$, which amounts to fixing the mode basis on which we express our states and does not affect their overlap.

Assume that $(\Gamma, \tilde\Gamma)$ constitute an extreme point of the target function. This implies that no infinitesimal change in the matrix $\tilde\Gamma$ will change the overlap, that is,
\begin{equation}
    \left.\dfrac{d}{dt}\det(\Gamma + e^{t(o\oplus q)}\tilde\Gamma e^{-t(o\oplus q)})\right|_{t=0} = 0,\quad  \forall o,q\in\mathfrak{o}(2n).
\end{equation}
By differentiating, and then using $\det(\Gamma + \tilde\Gamma)>0$ (since we are looking for maxima) and the cyclicity of the trace, we arrive at
\begin{align}
    \det(\Gamma + \tilde\Gamma)\tr{\left((\Gamma + \tilde\Gamma)^{-1}[o\oplus q, \tilde\Gamma]\right)} &= 0,\nonumber\\
    \tr{\left((o\oplus q), [\tilde\Gamma, (\Gamma + \tilde\Gamma)^{-1}]\right)} &= 0.\label{comm}
\end{align}
Let $[\tilde\Gamma, (\Gamma + \tilde\Gamma)^{-1}]$, which is skew-symmetric, have the following block structure (according to the bipartition of our states),
\begin{equation}
    [\tilde\Gamma, (\Gamma + \tilde\Gamma)^{-1}]\equiv \left(\begin{array}{cc}
        A & B \\
        -\trans B & D
    \end{array}\right),
\end{equation}
with $A = -\trans A, D = -\trans D$. Then condition \eqref{comm} implies 
\begin{equation}
    \tr{\left(\left(\begin{array}{cc}
        o & 0 \\
        0 & q
    \end{array}\right)\left(\begin{array}{cc}
        A & B \\
        -\trans B & D
    \end{array}\right)\right)} =\tr{(oA+qD)}= 0,
\end{equation}
which holds for every skew-symmetric $o,q$. This forces $A = D = 0$ and we conclude 
\begin{equation}
    [\tilde\Gamma, (\Gamma + \tilde\Gamma)^{-1}] = \left(\begin{array}{cc}
        0 & B \\
        -\trans B & 0
    \end{array}\right).
    \label{introB}
\end{equation}
We denote $\mathcal{B}\equiv[\tilde\Gamma, (\Gamma + \tilde\Gamma)^{-1}] = -[\Gamma, (\Gamma + \tilde\Gamma)^{-1}]$ for brevity. We proceed by noting
\begin{align}
    \{\Gamma, [\Gamma, (\Gamma + \tilde\Gamma)^{-1}]\} &= [\Gamma^2, (\Gamma + \tilde\Gamma)^{-1}] \\
    &= [-\mathds{1},(\Gamma + \tilde\Gamma)^{-1}] = 0,
\end{align}
where $\{,\}$ denotes the anticommutator. Thus
\begin{equation}
    \{\Gamma, \mathcal{B}\} = 0,
    \label{anticomm}
\end{equation}
which further constrains the form of $B$. Indeed, because of our assumption on $O,Q$, we have
\begin{equation}
    \Gamma = \left(\begin{array}{cc}
        \Gamma_{11} & \Gamma_{12} \\
        -\trans\Gamma_{12}  & \Gamma_{22}
    \end{array}\right),
\end{equation}
where
\begin{align}
    \Gamma_{11} &= \bigoplus_{i=1}^{n}\cos\theta_i\,J = -\Gamma_{22},\\
    \Gamma_{12} &= \bigoplus_{i=1}^{n}\sin\theta_i\,\mathds{1}.
\end{align}
Condition \eqref{anticomm} can be seen to imply
\begin{equation}
    [\Gamma_{11}, B] = 0, \quad \Gamma_{12}\trans B = - B \Gamma_{12}, \quad \Gamma_{12} B = - \trans B  \Gamma_{12},
\end{equation}
which, thanks to our assumptions about the $\theta_i$, is enough to force
\begin{equation}
    B = \bigoplus_{i=1}^n{b_i J},
\end{equation}
for some $b_i\in\mathbb{R}$. Now we go back to \eqref{introB} and write
\begin{align}
    [\tilde \Gamma, (\Gamma + \tilde\Gamma)^{-1}] &= \tilde \Gamma (\Gamma + \tilde\Gamma)^{-1} - (\Gamma + \tilde\Gamma)^{-1} \tilde \Gamma \nonumber \\
    &= \tilde \Gamma (\Gamma + \tilde\Gamma)^{-1} - \Gamma (\Gamma + \tilde\Gamma)^{-1} \nonumber \\
    &=(\tilde \Gamma - \Gamma) (\Gamma + \tilde\Gamma)^{-1},
\end{align}
hence,
\begin{align}
    \mathcal{B}(\Gamma + \tilde\Gamma) &= (\tilde \Gamma - \Gamma) \implies \nonumber \\
    \implies\tilde\Gamma &= (\mathds{1} - \mathcal{B})^{-1} \Gamma (\mathds{1} - \mathcal{B}),
\end{align}
where the inverse of $\mathds{1} - \mathcal{B}$ exists since $\det(\mathds{1} - \mathcal{B}) = \det(\mathds{1}+\trans B B)>0$. Defining $\beta_i \equiv 2\arctan{b_i}$, the expression above yields
\begin{align}
    \tilde\Gamma = \bigoplus_{i=1}^n{W(\theta_i + \beta_i)},
\end{align}
which can be checked to be consistent with all the conditions derived before, in particular
\begin{equation}
    [W(\theta+\beta),(W(\theta)+W(\theta+\beta))^{-1}] = \tan\frac{\beta}{2}\left(\begin{array}{cc}
       0 & J\\
       J & 0
    \end{array}\right).
\end{equation}
In conclusion, the pairs $\Gamma, \tilde\Gamma$ with maximal overlap for fixed spectra satisfy that $\Gamma$ and $\tilde\Gamma$ are simultaneously singular-value-decomposable, by which we mean there exists a basis in which
\begin{equation}
    \Gamma = \bigoplus_{i=1}^n{W(\theta_i)},\qquad \tilde\Gamma = \bigoplus_{i=1}^n{W(\tilde\theta_{\sigma(i)})},
\end{equation}
for some permutation $\sigma$. The statement of the theorem then follows from simply computing the overlap of these states, and extends to the case of general $\{\theta_i,\tilde\theta_i\}$ by a continuity argument.
\end{proof}
The case described in Lemma \ref{FGRA} follows as a corollary to the previous theorem by forcing all but $r$ of the $\tilde \theta_i$ to be equal to 0. It can then be seen that the optimal choice for the remaining ones is for them to equal the $r$ largest $\theta_i$ (the most entangled pairs) and for the permutation $\sigma$ to match them accordingly, so that the maximum overlap is given by \eqref{maxoverlap}, once we express the Gaussian entanglement spectrum back in terms of $\lambda_j$. The bound is tight since an approximation with such an overlap can be obtained by performing the Gaussian SVD of the target state and setting all but the $r$ largest $\theta_i$ to 0 (i.e., all but the $r$ smallest $|\lambda_j|$ to 1).

\section{Proof of Lemma \ref{property}}
\label{app:thm_proof}
As we advanced in the main text, we begin by proving a corresponding result in the thermodynamic limit:

\begin{lemma}
For the ground state of \eqref{Ham} on an \emph{infinite} chain, let $\mathcal{I}_{L,\infty}(\mu)$ be the number of eigenvalues $\lambda$ from the Gaussian entanglement spectrum of an interval of size $L$ that satisfy $|\lambda|<\mu$, and let $c>0$. Then there exists $\mu<1$ such that
\begin{equation}
    \mathcal{I}_{L,\infty}(\mu) > c\log L,
\end{equation}
as $L\to\infty$.
\label{property_inf}
\end{lemma}

\begin{proof}
Let $C_{L,\infty}$ be the correlation matrix of the interval of length $L$, and $V_{L,\infty}\equiv 2C_{L,\infty}-\mathds{1}$. Call $D_L(z)\equiv \det(z\mathds{1} - V_{L,\infty})$, and let $f(z)$ be a holomorphic function on a domain that includes the interval $[-1,1]$ where all the eigenvalues $\{\lambda_i\}$ of $V_{L,\infty}$ lie. Since we have
\begin{equation}
    D_L(z) = \prod^L_{i = 1}{(z - \lambda_i)},
\end{equation}
it follows from Cauchy's integral formula that
\begin{equation}
    \sum_{i=1}^L{f(\lambda_i)} = \dfrac{1}{2\pi i}\int_{\mathcal C}{dz f(z)\dfrac{d}{dz}\log D_L(z)},
    \label{complexint}
\end{equation}
where $\mathcal C$ is a contour within the domain of holomorphicity of $f$ encircling the interval $[-1,1]$. Since $V_{L,\infty}$ is a Toeplitz matrix with an adequate symbol, the asymptotic value of $D_L(z)$ as $L\to\infty$ is given to us by the \textit{Fisher-Harwig conjecture}, in particular by a subcase thereof which was proven by Basor \cite{Basor79}. This property has been exploited for various computations in the XX model \cite{Jin04}. In our particular case, it tells us 
\begin{align}
    \log D_L(z)&\sim L\log\sqrt{z^2-1} \nonumber\\ &+ \dfrac{\log L}{2\pi^2}\left[\log{\left(\dfrac{z+1}{z-1}\right)}\right]^2 
\end{align}
where by $\sim$ we mean both sides are equal up to $O(1)$ terms that do not grow with $L$. The right hand side of \eqref{complexint} then reads
\begin{align}
    &\oint{dz f(z)\dfrac{d}{dz}\log D_L(z)}\sim\\
    &\qquad\sim \dfrac{L}{2} \oint{dz f(z)\left(\dfrac{1}{z-1}+\dfrac{1}{z+1}\right)}\\
    &\qquad+\dfrac{2\log L}{\pi^2} \oint{dz f(z)\log\left(\dfrac{z-1}{z+1}\right)\dfrac{1}{z^2-1}}
\end{align}
Using complex variable techniques, this finally results in 
\begin{equation}
    \sum_{i=1}^L{f(\lambda_i)} \sim L\dfrac{f(-1)+f(1)}{2}+\dfrac{2\log L }{\pi^2}\int^{1}_{-1}{d\lambda\dfrac{f(\lambda)}{1-\lambda^2}},
\end{equation}
which is a strong constraint on the distribution of eigenvalues. It hints at the fact that they are asymptotically distributed with a density $2\log L/(\pi^2(1-\lambda^2))$ along the interval $[-1,1]$, with the rest of them (a number of order $L$) eventually clumping at the endpoints.
We are now in a position to bound the function $\mathcal I_{L,\infty}(\mu)$. It can be written as a sum over eigenvalues, with $f$ the indicator function of the interval $[-\mu, \mu]$, which of course cannot be extended to a holomorphic function. Still, to get intuition, the result would be
\begin{equation}
    \mathcal I_{L,\infty}(\mu)~\text{``}=\text{''}~\dfrac{2\log L }{\pi^2}\int^{\mu}_{-\mu}{d\lambda\dfrac{1}{1-\lambda^2}} = \dfrac{4\log L \argtanh\mu}{\pi^2},
\end{equation}
and since the coefficient of $\log L$ diverges as $\mu\to 1$, we would have the result. To make a proper statement, we use the functions
\begin{equation}
    f_\mu(\lambda) \equiv \dfrac{(1-\lambda^2)(\mu^2-\lambda^2)}{(2-\mu^2-\lambda^2)^2},
\end{equation}
which lower bound the indicator function $\Theta(\mu-|\lambda|)$ and are holomorphic on a disk containing $[-1,1]$. Thus we can assure,
\begin{align}
    \mathcal{I}_{L,\infty}(\mu) \geq \sum_{i=1}^L{f_\mu(\lambda_i)} \sim \dfrac{2\log L }{\pi^2}\int^{1}_{-1}{d\lambda\dfrac{f_{\mu}(\lambda)}{1-\lambda^2}}~~~~~~~~~~\nonumber\\
    =\dfrac{4\log L }{\pi^2}\left(\dfrac{\argtanh{\left(\dfrac{1}{\sqrt{2-\mu^2}}\right)}}{(2-\mu^2)^{3/2}}-\dfrac{1}{2-\mu^2}\right)
\end{align}
and since the coefficient of $\log L$ on the rhs still diverges as $\mu\to 1$, the result follows.
\end{proof}

Now we will show that the spectra of the correlation matrices for the finite and infinite chains are close enough that Lemmas \ref{property} and \ref{property_inf} imply each other. Denote the Frobenius norm by $\|\cdot\|_2$ and the Schatten 1-norm (or trace norm) by $\|\cdot\|_1$. We then have,
\begin{lemma}
Let $C_{L,N},C_{L,\infty}$ be the correlation matrices for an interval of $L$ sites of a finite chain of $N$ sites and an infinite chain, respectively, and let $L/N = \varphi$ stay constant as we increase $N$. Then $\|C_{L,N}-C_{L,\infty}\|_1$ is bounded by a constant.
\label{norm_bound}
\end{lemma}
\begin{proof}
Both $C_{L,N}$ and $C_{L,\infty}$ are Toeplitz matrices. Their matrix elements read
\begin{align}
    \left(C_{L,N}\right)_{i, i+r} = \dfrac{1}{N}\dfrac{\sin{\left(\dfrac{\pi}{2}r+\dfrac{m}{N}r\right)}}{N\sin{\dfrac{\pi r}{N}}\phantom{\Big|}},\label{fincorr}\\
    \left(C_{L,\infty}\right)_{i, i+r} = \dfrac{1}{N}\dfrac{\sin{\left(\dfrac{\pi}{2}r\right)}}{\pi r},
\end{align}
where $m = 2,1,0,-1$ whenever $N \equiv 0,1,2,3 \mod 4$ respectively. Define $L\times L$ Toeplitz matrices $T^{\text{even}}_j, T^{\text{odd}}_j$ with elements
\begin{align}
    \left(T^{\text{even}}_j\right)_{i,i+r} &\equiv \cos\left(\dfrac{\pi r}{2}\right) r^{2j},\\ \left(T^{\text{odd}}_j\right)_{i,i+r} &\equiv \sin\left(\dfrac{\pi r}{2}\right) r^{2j+1}.
\end{align}
By expanding and collecting terms cautiously in \eqref{fincorr}, it can be seen that 
\begin{equation}
    C_{L,N} - C_{L,\infty}  =  \sum_{j=0}^\infty{\dfrac{a_j}{N^{2j+1}}T^{\text{even}}_j} + \sum_{j=0}^\infty{\dfrac{b_j}{N^{2j+2}}T^{\text{odd}}_j},
    \label{diffexpansion}
\end{equation}
where $a_j, b_j$ are the coefficients in the series expansion of the holomorphic functions
\begin{align}
    \dfrac{\sin{mz}}{\sin{\pi z}} = \sum^\infty_{j=0}{a_jz^{2j}},\\
    \dfrac{\cos{mz}}{\sin{\pi z}} - \dfrac{1}{\pi z} = \sum^\infty_{j=0}{b_jz^{2j+1}},
\end{align}
which are absolutely summable within their disc of convergence (the unit disc). The trace norm of $T^{\text{even}}_j, T^{\text{odd}}_j$ can be bounded by using
\begin{equation}
    \text{rank}(T^{\text{even}}_j) = 4j+2,\qquad \text{rank}(T^{\text{odd}}_j) = 4j+4,
\end{equation}
together with the inequality
\begin{equation}
    \|M\|_1\leq\sqrt{\text{rank}(M)}\|M\|_2,
\end{equation}
to find
\begin{align}
    \|T^{\text{even}}_j\|_1&\leq\sqrt{4j+2}\|T^{\text{even}}_j\|_2\nonumber\\
    &\leq\sqrt{4j+2}\sqrt{2\sum^{L}_{\substack{{r=0}\\r\text{ even}}}{r^{4j}(L-r)}}\nonumber\\
    &\leq\sqrt{4j+2}\sqrt{2L\sum^{L}_{\substack{{r=0}\\r\text{ even}}}{r^{4j}}}\nonumber\\
    &\leq\dfrac{\sqrt{4j+2}}{\sqrt{4j+1}}L^{2j+1},
\end{align}
and 
\begin{align}
    \|T^{\text{odd}}_j\|_1&\leq\sqrt{4j+4}\|T^{\text{even}}_j\|_2\nonumber\\
    &\leq\sqrt{4j+2}\sqrt{2\sum^{L}_{\substack{{r=1}\\r\text{ odd}}}{r^{4j+2}(L-r)}}\nonumber\\
    &\leq\sqrt{4j+2}\sqrt{2L\sum^{L}_{\substack{{r=1}\\r\text{ odd}}}{r^{4j+2}}}\nonumber\\
    &\leq\dfrac{\sqrt{4j+4}}{\sqrt{4j+3}}L^{2j+2}.
\end{align}
Going back to \eqref{diffexpansion} this yields
\begin{align}
     \|C_{L,N} - C_{L,\infty}\|_1 &\leq  \sum_{j=0}^\infty{\dfrac{|a_j|}{N^{2j+1}}\|T^{\text{even}}_j}\|_1 +\dfrac{|b_j|}{N^{2j+2}}\|T^{\text{odd}}_j\|_1,\nonumber\\
     &\leq 2\sum_{j=0}^\infty{|a_j|\left(\dfrac{L}{N}\right)^{2j+1} +|b_j|\left(\dfrac{L}{N}\right)^{2j+2}},
\end{align}
which converges and is thus bounded as $N\to\infty$ with constant $L/N$.
\end{proof}

Finally, we have

\begin{proof}[Proof (of Lemma \ref{property})]
We will argue by contradiction. Assume therefore that there is $c>0$ such that for all $\mu<1$, ${\mc I_{L,N}(\mu)\leq c\log L}$. Since Lemma \ref{property_inf} holds, we can choose $\mu < \mu' <1$ such that 
\begin{equation}
    \mc I_{L,\infty}(\mu)>(c+1)\log L \geq \mc I_{L,N}(\mu') + \log L.
\end{equation}
Recall now the following inequality for the trace norm of the difference of Hermitian matrices,
\begin{equation}
    \sum_i{|\alpha_i-\beta_i|} \leq \|A-B\|_1,
\end{equation}
where $\alpha_i, \beta_i$ are the eigenvalues of $A,B$ in descending order \cite{Bhatia96}. We choose $A = V_{L,N}\equiv 2C_{L,N}-\mathds{1}, B = V_{L,\infty}\equiv 2C_{L,\infty}-\mathds{1}$. In our situation, there are asymptotically at least $\log L$ eigenvalues of $A$ that are at least $\mu'-\mu$ away from their associated eigenvalues of $B$, thus the left hand side of the inequality grows with $L$ while the right hand side is bounded by Lemma \ref{norm_bound}, a contradiction.
\end{proof}

\section{Numerical methods}
\label{app:num}
Here we present some numerical results for the approximation of the ground state of our hopping model \eqref{Ham} with GfMPS. After a few generic optimizations within the generic GfMPS class, we found that the numerical optima always fell within a particular subclass of GfMPS, which we dub \textit{ladder GfMPS}, and introduce in what follows. 

To begin with, we recall the basics of GfMPS contration in momentum space (we stay at one fermionic orbital per site, the generalization to a higher number thereof is straightforward). Once more, it is convenient to work in the Majorana representation, where two Hermitian operators $c_{j,1},c_{j,2}$ stand for each fermion mode $\phdg a_j, \dg a_j$,
\begin{equation}
    c_{j,1}\equiv \phdg a_j + \dg a_j,\qquad c_{j,2}\equiv i(\phdg a_j-\dg a_j).
\end{equation}
A Fourier transform then maps these to complex Majorana operators $d_{k,1},d_{k,2}$,
\begin{equation}
    d_{k,s} = \dfrac{1}{N}\sum_{j=1}^N{e^{-ikj}c_{j,s}},\qquad s=1,2.
\end{equation}
This is useful in the translation invariant case, for which different momenta decouple. At the level of correlation matrices, this implies that the Majorana correlation matrix $\Gamma$,
\begin{equation}
    [\Gamma]_{js,j's'} \equiv \dfrac{i}{2}\ev{[c_{j,s},c_{j',s'}]},
\end{equation}
is block diagonalized by the Fourier transform $\mathcal{F}$,
\begin{equation}
    \mathcal{F}\Gamma\mathcal{F}^\dagger = \bigoplus_{k}{G(k)},\quad [G(k)]_{ss'}\equiv \dfrac{i}{2}\ev{[\phdg d_{k,s},\dg d_{k, s'}]}.
\end{equation}
Our translation invariant GfMPS will be determined by a fiducial state of 2 Majorana fermions, $\chi$ left virtual Majorana fermions and $\chi$ right virtual Majorana fermions. Note that in this Appendix $\chi$ differs by a factor of 2 from $\chi$ in the main text, since it counts the number of virtual Majorana fermions. Because we are working with periodic boundary conditions, we can allow $\chi$ to be odd. In fact, the parity of $\chi$ can have significant consequences for the parity structure of the states in the variational class \cite{Mortier20}, and in our case, odd $\chi$ is actually preferable. We denote the correlation matrix of the fiducial state by
\begin{equation}
    \Gamma = \left(\begin{array}{cc}
       A & B \\
        -B^T & D
    \end{array}\right),
\end{equation}
where the block structure comes from separating the two physical fermions ($A$ is a $2\times 2$ submatrix) from the virtual fermions ($D$ is a $2\chi\times2\chi$ submatrix). The correlation matrix for the GfMPS state is obtained by projecting the virtual Majorana fermions onto entangled pairs, which in momentum space reads \cite{Kraus10}
\begin{equation}
    G(k) = A + B\left[D -\left(\begin{array}{cc}
        0 & e^{ik}\mathds{1} \\
        -e^{-ik}\mathds{1} & 0
    \end{array}\right)\right]^{-1}B.
\end{equation}

Next we define a \textit{rail GfMPS}, which is characterized by a $(f+\chi)\times(f+\chi)$ orthogonal matrix $O$ that is divided in blocks
\begin{equation}
    O =\left(\begin{array}{cc}
       O_{11} & O_{12}\\
       O_{12} & O_{22}
    \end{array}\right),
    \label{Oblocking}
\end{equation}
where $O_{11}$ is $f\times f$ and $O_{22}$ is $\chi\times\chi$. The correlation matrix for the fiducial state of $2f$ physical fermions and $2\chi$ virtual fermions for the rail GfMPS is defined to be
\begin{equation}
    \Gamma_O \equiv\left(\begin{array}{cccc}
       0 & O_{11} &0 & O_{12} \\
       -\trans O_{11} & 0 & -\trans O_{21} & 0\\
        0& O_{21} &0 & O_{22}\\
       -\trans O_{12} & 0& -\trans O_{22}  & 0       
    \end{array}\right).
\end{equation}
Therefore, the $2f\times2f$ correlation matrix $G(k)$ for the resulting GfMPS state is
\begin{equation}
    G(k)= \left(\begin{array}{cc}
       0 & T(e^{ik})\\
       -\dg {T(e^{ik})} & 0
    \end{array}\right),
\end{equation}
where
\begin{equation}
    T(z) = O_{11} + O_{12}(O_{22}-z\mathds{1})^{-1}O_{21}
\end{equation}
is a unitary matrix, or, in our $f=1$ case, a complex phase. In fact, readers familiar with the theory of discrete linear time invariant (LTI) systems may recognize $T(z)$ as the transfer function of a lossless system whose state space representation is given by the matrix $O$ with the blocking from \eqref{Oblocking}. This analogy could be exploited to import techniques from the LTI system literature to the GfMPS setting. Here we will not pursue it further. It is nevertheless known that $T(z)$ will be a \textit{finite Blaschke product}, i.e. a unimodular rational function \cite{Alpay16}, of the form
\begin{equation}
    T(z) = \eta\prod_{j=1}^\chi{\dfrac{1-\bar \alpha_j z}{z-\alpha_j}},
\end{equation}
where $|\eta|=1$ and $\alpha_j$ are the eigenvalues of $O_{22}$, which can be any conjugation invariant set of complex numbers inside the unit disk \footnote{Lacking a more intuitive proof, this can be seen as follows: let $\{\alpha_j\}$ be the desired set of eigenvalues. The constraint that $O_{22}$ is a submatrix of an orthogonal matrix may be rephrased in terms of its singular values, demanding all but one of them to be 1 (the latter must equal $\prod_j{|\alpha_j|}$). Then $O_{22}$ can be found as the matrix whose eigenvalues and singular values are those we just prescribed \cite{Li01}, and then extended to an orthogonal matrix $O$.}.

\begin{figure}[t]
	\centering
	\includegraphics[width=0.8\linewidth]{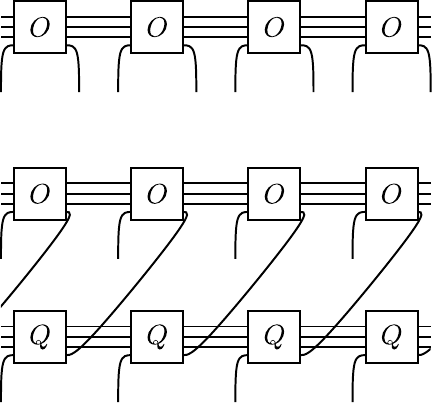}
	\caption{(top) An $f=1,\chi = 3$ rail GfMPS. (bottom) An $f=1, \chi = 7$ ladder GfMPS.}
	\label{fig:rail_ladder}
\end{figure}

An $f=1$ \textit{ladder GfMPS} is made from juxtaposing two $f=1$ rail GfMPS and projecting their respective second physical Majorana fermions onto maximally entangled pairs, to form the ``rungs'' of the ladder (see Figure \ref{fig:rail_ladder}). The resulting GfMPS has again $f=1$ and a bond dimension that equals the sum of those of the rails plus one (from the rungs), $\chi = \chi_1+\chi_2+1$. Its correlation matrix is given by
\begin{equation}
    G(k)= \left(\begin{array}{cc}
       0 & e^{ik}T_1(e^{ik})T_2(e^{ik})^*\\
       -e^{-ik}{T_1(e^{ik})}^*T_2(e^{ik}) & 0
    \end{array}\right)
\end{equation}
What was gained from this construction is that the product $T_1(e^{ik})T_2(e^{ik})^*$ is now a unimodular rational function with poles no longer confined to the unit disk, and thus more general. It can be written in terms of an arbitrary polynomial and its reciprocal polynomial, and a few additional manipulations lead to the general form of $n_k$ we showed in the main text,
\begin{equation}
    n_k = \dfrac{p(\cos \frac{k}{2})^2}{p(\cos \frac{k}{2})^2+q(\sin \frac{k}{2})^2},
\end{equation}
for $p,q$ arbitrary real odd monic polynomials of degree $\chi$ (assuming $\chi$ is odd), or equivalently,
\begin{equation}
    n_k = \dfrac{(1 + \cos k)\,\pi(\cos k)^2}{(1 + \cos k)\,\pi(\cos k)^2+(1 - \cos k)\,\theta(\cos k)^2},
    \label{generic_nk_app}
\end{equation}
for $\pi,\theta$ arbitrary real monic polynomials of degree $\frac{\chi-1}{2}$. We can then try to guess adequate families of polynomials that make $n_k$ close to its exact value from Eq. \eqref{corrmat} on the allowed momenta $k\in\frac{2\pi}{N}\mathbb{Z}\cap(-\pi,\pi]$. We tried expressions based on Fourier expansions of the exact $n_k$ and on Chebyshev polynomials, which nevertheless displayed a clearly exponential b.d. Our best results came from picking $p$ (resp.~$q$) so that its zeros are exactly a subset of the allowed momenta that are outside (resp.~inside) the Fermi surface. For those selected values, the GfMPS approximation reproduces the target state exactly. Several approaches can be followed to choose which precise momenta we make exact. Choosing all of them next to the Fermi points still leads to exponential b.d., but spreading them logarithmically (so that we still pick exponentially more momenta that are close to the Fermi points), leads to a very well-behaved ansatz family that gives rise to the results shown in the main text.

\bibliography{main}
\bibliographystyle{unsrt}

\end{document}